\documentclass[a4paper,UKenglish]{lipics-v2018}

\usepackage{microtype}
\usepackage{graphicx,color,amssymb,comment,amsmath,soul,tikz,setspace,float,dirtytalk}
\usepackage[algoruled]{algorithm2e}
\usepackage{hyperref}
\usetikzlibrary{positioning,arrows,chains,matrix,positioning,scopes} 
\tikzset{main node/.style={circle,fill=black!10,draw,minimum size=0.5cm,inner sep=1pt},}

\newcommand{\cleft}[0]{C_{l}}
\newcommand{\cright}[0]{C_{r}}

\theoremstyle{plain}
\newtheorem{problem}[theorem]{Problem}

\EventEditors{Wen-Lian Hsu, Der-Tsai Lee, and Chung-Shou Liao}
\EventNoEds{3}
\EventLongTitle{29th International Symposium on Algorithms and Computation (ISAAC 2018)}
\EventShortTitle{ISAAC 2018}
\EventAcronym{ISAAC}
\EventYear{2018}
\EventDate{December 16--19, 2018}
\EventLocation{Jiaoxi, Yilan, Taiwan}
\EventLogo{}
\SeriesVolume{123}
\ArticleNo{<article-no>} 

\title{Algorithmic Blockchain Channel Design}

\author{Georgia Avarikioti}{ETH Zurich, Switzerland}{zetavar@ethz.ch}{}{}
\author{Yuyi Wang}{ETH Zurich, Switzerland}{yuwang@ethz.ch}{}{}
\author{Roger Wattenhofer}{ETH Zurich, Switzerland}{wattenhofer@ethz.ch}{}{}

\authorrunning{G. Avarikioti, Y. Wang, R. Wattenhofer.}

\Copyright{Georgia Avarikioti and Yuyi Wang and Roger Wattenhofer}

\subjclass{}


\nolinenumbers 

\begin{document}

\maketitle

\begin{abstract}
Payment networks, also known as channels, are a most promising solution to the throughput problem of cryptocurrencies. In this paper we study the design of capital-efficient payment networks, offline as well as online variants. We want to know how to compute an efficient payment network topology, how capital should be assigned to the individual edges, and how to decide which transactions to accept. Towards this end, we present a flurry of interesting results, basic but generally applicable insights on the one hand, and hardness results and approximation algorithms on the other hand.
\end{abstract}

\section{Introduction}
Cryptocurrencies such as Bitcoin \cite{nakamoto2008bitcoin} or Ethereum \cite{ethereum} have a serious throughput problem \cite{croman2016scaling}. They can process tens of transactions per second, whereas non-blockchain systems (credit card companies, inter-banking payment systems, paypal, etc.) can handle tens of \emph{thousands} of transactions per second. Various proposals have been made in an attempt to solve this throughput problem, e.g.,  sharding \cite{luu2016sharding,kokoris2017omniledger} or sidechains \cite{back2014sidechains}. However, payment networks (also known as channels) \cite{DW2015channels,poon2015lightning,raiden2017} are widely accepted to be the most promising of these so-called ``layer 2'' solutions, since payment networks allow data to go off-chain securely.

Duplex micropayment channels \cite{DW2015channels}, Lightning \cite{poon2015lightning} or Raiden \cite{raiden2017} are fast and scalable payment networks, where transactions between two users are executed in off-chain two-party channels. The blockchain is involved when opening a channel, as the foundation of a channel must be registered with the blockchain. In exceptions, if the two parties of a channel are in disagreement, the blockchain may also be involved as a safety net when closing a channel.

While the efficiency of channels is undisputed, payment networks have a reputation to be capital hungry and as such difficult to deploy.
In this paper we want to better understand this demand for capital, studying the issue from an algorithmic perspective. We want to know the complexity an operator of a payment network, a Payment Service Provider (PSP), will face when setting up a payment network.
%
%
\subsection{From Payment Channels to Network Design}

Consider a PSP wants to create a payment network. The PSP can open a channel between any two parties; technically this can be achieved using multi-party channels \cite{burchert2017multipartychannels}, where the two parties and the PSP join a three-party channel funded only by the PSP. 

Algorithmically speaking, a payment network is a graph, where each undirected edge $(u,v)$ is a payment channel between the parties $u,v$. When a channel (an edge) is established, PSP capital is locked into the channel on each side of the edge. 
This capital can then be moved on the channel, from $u$ to $v$ or vice versa, much like moving tokens from one side of an abacus to the other. For example, if initially a capital of 5 is locked on each side of the  $(u,v)$ channel, then a transaction with a value of 2 from $u$ to $v$ will reduce the capital on $u$'s side to 3, and increase the capital on $v$'s side to 7. Transactions can also be multi-hop, moving capital on each edge of the path, in the direction of the path of the transaction. The only constraint is that the capital on any side of any edge must be non-negative at all times.

The PSP needs to decide how to design the network, i.e., which edges (channels) the PSP should establish. Moreover, the PSP needs to decide how much capital it should assign to these newly established edges, in particular how much capital on each side of every edge.


Establishing a new channel not only involves capital (which is going to be reclaimed eventually), but will also cost (since each newly established channel needs to be registered with the blockchain). We model this channel opening cost as a constant, given that the fee the blockchain asks is (more or less) constant. The total cost is then the number of open channels (the edges of the network) times this constant cost to open each channel.


Our goal is to define a strategy for the PSP regarding which transactions to execute in order to maximize profit (fees from transactions minus costs to set up channels) and minimize capital (cryptomoney that is temporarily locked into channels). Note that there is a trade-off between profit and capital, as more capital may allow to accept more transactions, earning fees for each transaction, hence increasing profit
In particular, we discuss the following questions: What is the minimum capital needed to be able to accept a given set of transactions? What is the maximum profit we can achieve with a given capital? These questions are at the heart of understanding the Pareto-nature of the trade-off between profit and capital in payment networks.
\subsection{Related Work}
Current work on payment channels has mainly focused on designing routing algorithms for the implemented decentralized payment networks, such as the Lightning \cite{poon2015lightning} and Raiden \cite{raiden2017} networks. 
Prihodko et al.\ \cite{prihodko2016flare} present Flare, an efficient routing algorithm for the Lightning network by collecting information on the network's local topology. Malavolta et al.\ \cite{malavolta2017silentwhispers} introduce the IOU credit network SilentWhispers  where they use landmark routing to discover multiple paths and multi-party computation to decide the amount of capital to be locked on each path. Roos et al.\ \cite{roos2018routing} propose SpeedyMurmurs, a routing algorithm for payment networks that uses embedding-based path discovery to find routes from sender to receiver. However, all these protocols assume a network structure created by the individuals participating in the network. The goal is to discover the network topology and possible routes from sender to receiver of every transaction. Our objective is to design the optimal network structure assuming a central authority, the PSP.

An active line of research on payment channels is the construction of secure and private systems that can act as payment hubs. Heilman et al.\ \cite{heilman2017tumblebit} propose a Bitcoin-compatible construction of a payment hub for fast and anonymous off-chain transactions through an untrusted intermediary. Green et al.\ \cite{green2017bolt} present Bolt (Blind Off-chain Lightweight Transactions) for constructing privacy-preserving unlinkable and fast payment channels. However, they do not analyze how expensive the construction of a payment hub is for a PSP. In this work, we answer the following questions: is a payment hub a good solution for a PSP? How much capital is required to build a payment hub compared to the capital of a capital-optimal network? These answers are highly relevant to the economic viability of a payment hub as a practical solution for payment networks, and ultimately whether payment networks can solve the eminent throughput problem of cryptocurrencies.

Our paper can be seen as a cryptocurrency variant of classic work on network design. It is as such somewhat related to fundamental work starting in the 1970s.
For example, Johnson et al.\ \cite{johnson1978networkdesign} prove that given a weighted undirected graph, finding a subgraph that connects all the original vertices and minimizes the sum of the shortest path weights between all vertex pairs, subject to a budget constraint on the sum of its edge weights is NP-hard. 
Another similar problem is the optimum communication spanning tree problem \cite{hu1974ocst}, 
whose input is a set of nodes, the distances and requests between them, and the goal is to find the spanning tree that minimizes the cost of communication (for each pair, the request multiplied by the sum of distance). 
Our channel design problem seems similar to these problems since the routing of a transaction matters, and our objective is to minimize the capital on the channels (like the original network design work wants to minimize the sum of the distances). However, in contrast with traditional network design, in payment networks the order of transactions matters, as the capital moves from one side of the channels to the other. 
Moving capital gives network design a surprising twist, as classic techniques do not work anymore. With the anticipated importance of payment networks, we believe one should have a fresh look at network design.
%
%
\subsection{Our Contribution}
We introduce an algorithmic framework for the channel network design problem. 
First, we study the offline problem, i.e., we are given the future sequence of transactions. We show that maximizing the profit given the capital assignments is NP-hard, even for a single channel. Then, we present a fully polynomial time approximation scheme for the single channel case. Later, we consider the case where the PSP wants to maximize its profit and thus execute all profitable transactions. We prove that a hub (a star graph) is a $2$-approximation with respect to the capital. Moreover, we show the problem is NP-complete under graph restrictions. 

In addition, we examine online variants. First, we examine the online single channel case assuming the PSP wants to maximize its profit under capital constraints. We show that there is no deterministic competitive algorithm for adaptive adversaries. Later, we study the online channel design problem assuming all profitable transactions are executed. We show that the star graph yields an $O(\log C)$-competitive algorithm, where $C$ denotes the optimal capital.

Omitted proofs are included in the Appendix. 
%
%
\section{Notation and Problem Variants}\label{sec:model}
We assume the fee of a transaction on the blockchain to be constant, without loss of generality simply $1$. The fee of a transaction in the payment network cannot be higher than the fee on the blockchain, or a potential user may prefer the blockchain over the payment network. A rational PSP will ask for a transaction processing fee which is as high as possible but lower than the blockchain fee, hence for $1-\epsilon$. In our analysis we will usually assume that $\epsilon \rightarrow 0$.

Let us now formally define the problems we will study.
\begin{problem}[General Payment Network Design]\label{def:capitalcon}
\leavevmode 

\emph{Input:} Capital $C$, profit $P$, the sequence of $n$ transactions $t_i=(s_i,r_i,v_i)$ with $1\leq i \leq n$, each containing the sender node $s_i$, the receiver node $r_i$ and the value $v_i$ of the transaction $t_i$. 

\emph{Output:} Strategy $S={\{0,1\}}^n$, a binary vector where the $i^{th}$ position is $1$ if we choose to execute the $i^{th}$ transaction of the input and $0$ else. The graph $G(V,E,\cleft,\cright)$ is the network we created to execute the chosen transactions, where $V$ is the set of senders and receivers that participate in any transaction, $E$ is the set of channels we open and $\cleft,\cright$ the capital on each side of each edge. Each transaction can be routed arbitrarily in $G$, denoted by $S_e={\{-1,0,1\}}^n, for \ all \ e \in E$, i.e., $S_e(i)=1$ (or $-1$) if transaction $i$ is routed through edge $e$ from left to right (from right to left, respectively) and $S_e(i)=0$ if transaction $i$ is not routed through edge $e$.

Our goal is to return (if it exists) a strategy $S$, a graph $G$ and a routing $S_e$
subject to the following constraints:
\begin{enumerate}
\item $|S| - |E| \geq P $
\item $\forall e \in E , \forall j \in \{1,2, \dots n\},~  -\cright(e) \leq \sum_{i=1}^j  S_e(i) \cdot v_i \leq \cleft(e)$
\item $\sum_{\forall e \in E}\cleft(e)+\cright(e)+|E| \leq C$
\end{enumerate}
\end{problem}

The first inequality guarantees that the fees of the accepted transactions minus the cost of opening the channels is at least as high as the intended profit. The second inequality makes sure that at any time the capital on each side of each channel is non-negative. The third inequality ensures that the used capital on the channels and the cost of opening the channels is at most the available capital.

Problem \ref{def:capitalcon} in all its generality is difficult, as it features many variables. Consequentially, we mostly focus on the most interesting special cases of Problem \ref{def:capitalcon}: We consider transactions on a single channel between just two nodes. And we consider minimizing the capital assuming all profitable transactions are executed. 
Formally the problems we examine are the following.

\begin{problem}[Single Channel]
\label{def:SingleEdge}
Given a sequence of $n$ transactions $t_i=(s,r,v_i)$, where $s$ and $r$ are the nodes of the single edge $e$, a capital assignment $\cright(e), \cleft(e)$, and a profit $P$, decide whether there is a strategy $S$ such that $|S|\geq P$ and $\forall j \in [n],~  -\cleft(e) \leq \sum_{i=1}^j  S(i) \cdot v_i \leq \cright(e)$.
\end{problem}

\begin{problem}[Channel Design for All Transactions]
\label{def:maxprofit}
Given a sequence of $n$ transactions $t_i=(s_i,r_i,v_i)$, return the graph $G(V,E)$ that achieves maximum profit with minimum capital $C$.
\end{problem}

\begin{problem}[Capital Assignment and Routing]
\label{def:subgraph}
Given a graph $G(V,E)$, a sequence of $n$ transactions $t_i=(s_i,r_i,v_i)$ and a capital $C$, determine whether all transactions can be executed in $G$ with the given capital $C$. 
\end{problem}
%
%
\section{Offline Channel Design}
In this section, we study the offline channels network design problem, i.e., we assume we know the future transactions (for the next period). First, we explore the network topology for the general problem. Then, we examine the case where we are given a specific capital (or even a capital assignment) and we aim to maximize the PSP's profit, hence execute as many transactions as possible. We focus on solving the problem for a single edge of the network, since even in this simple case the problem is challenging. Later, we focus on minimizing the capital given the PSP wants to execute all the profitable transactions.
%
%
\subsection{Graph Topology}
We first prove some observations concerning the optimal graph structure. We consider as optimal the 
solution that maximizes the profit while respecting the capital constrains (optimization version of Problem \ref{def:capitalcon}).
\begin{lemma}\label{lem:twotx}
The graph of the optimal solution does not contain any node that sends and receives less than two transactions.
\end{lemma}
Thus, during preprocessing we can safely remove all transactions that contain a node that is only sender or receiver of a transaction in this one transaction. The time complexity of this procedure is linear in the number of transactions.
\begin{lemma}\label{lem:tree}
The optimal graph is not necessarily a tree (or forest).
\end{lemma}

Due to the complexity of the problem we focus on a single channel. It turns out that even for this degenerate case, the problem is far from trivial.

\subsection{Single Channel}
We now focus on a single channel. We prove that even in this case the problem of choosing the transactions that maximize the profit given capital assignments is NP-hard and present an FPTAS.

Specifically, we are given a sequence of transactions on a single edge of a network and their values, the capital assignment on the edge and a target profit. Our goal is to decide whether we can execute at least as many transactions as the given target profit while respecting the capital constraints. Since the number of edges is fixed and equal to $1$ the profit now is the number of executed transactions (Problem \ref{def:SingleEdge}). The problem is equivalent to a variant of the $0/1$ knapsack problem where each transaction represents an item. Each item has profit $1$ and either positive or negative size (values). The capacity of the knapsack is represented by the capital assignments and the goal is to maximize the profit while respecting the capacity.

\begin{problem}[Fixed Weight Subset Sum (\textsf{FWSS})]
Given a set of non-negative integers $U = \{a_1, a_2, \ldots, a_n \}$, and non-negative integers $A$ and $l$, is there a non-empty subset $U' \subseteq U$ such that $|U'|=l$ and $\sum_{a_i \in U'} a_i = A$? 
\end{problem}
\begin{lemma}\label{lem:FWSS}
\textsf{FWSS} is NP-hard.
\end{lemma}

\begin{theorem}
Problem \ref{def:SingleEdge} is NP-hard.
\end{theorem}
\begin{proof} 
We will reduce Fixed Weight Subset Sum (\textsf{FWSS}) to Problem \ref{def:SingleEdge}. 

Assuming we are given an instance of the \textsf{FWSS}, we present a polynomial time transformation to an instance of Problem \ref{def:SingleEdge}. 
We first define the capital assignment on the edge $\cright(e)=A(l+1)$, $\cleft(e)=0$ and the profit $P=l+n(l+1)$. 
Then, we define the sequence of transactions as follows: $v_i=a_i+A $, $\forall 1 \leq i \leq n$ and $v_i= -A/n$, $\forall n < i \leq n(l+2)$. We will prove that there is a non-empty set that satisfies the \textsf{FWSS} problem if and only if we can choose transactions that satisfy the capital constraints and profit in the aforementioned instance. 

Assume we have a "yes" instance of the problem. Then, we have chosen at least $P=l+n(l+1)$ transactions to execute. We will show that this corresponds to choosing $l$ 
positive transactions that sum up to $A(l+1)$, thus to a solution of the \textsf{FWSS} problem. Towards contradiction, we examine the following three cases:
\begin{itemize}
\item If the number of positive transactions is less than $l$, the total profit is less than $l+n(l+1)$, since there are only $n(l+1)$ negative transactions. 
\item If the number of positive transactions is more than $l$, then we violate the capital constraints, since $\sum_{i} v_i \geq A(l+1) + \sum_{i} a_i > A(l+1)=\cright(e)$, where $i$ corresponds to the chosen transactions. 
\item Suppose the $l$ chosen transactions' values sum to more than $A(l+1)$. Then, the capital constraint is violated.
\item Suppose the $l$ chosen transactions' values sum to less than $A(l+1)$; suppose the sum is $Al+\sigma$ with some $\sigma<A$. 
Then, then negative transactions to be executed can be at most $\frac{lA}{A/n}+\frac{\sigma}{A/n}<ln+n$. Thus, the profit is strictly less than $l+ln+n$. Contradiction. 
\end{itemize}
Thus, a "yes" instance of our problem implies a "yes" instance of the \textsf{FWSS} problem.
For the other direction, we will prove that if there is no subsequence of transactions of size at least $P$ that satisfies the capital constraints, then  there is no subset of size $l$ that sums to $A$ in \textsf{FWSS}. Equivalently, we will show that if there is a subset of size $l$ that sums to $A$ in \textsf{FWSS}, then there exists a subsequence of transactions of size at least $P$ that satisfies the capital constraints. Suppose there is a non-empty set $U'\subseteq U$ such that $|U'|=l$ and $\sum_{a_i \in U'} a_i = A$. Then we can execute the $l$ transactions that correspond to the chosen $a_i$'s with exactly the $\cright(e)$ capital, which will be transfered on $\cleft(e)=A(l+1)$. Then, we can execute all the negative transactions since they are $n(l+1)$ many with values $A/n$, thus we need $A(l+1)=\cleft(e)$ capital. Therefore, we can execute $P=l+n(l+1)$ transactions, achieving the required profit while satisfying the capital constraints.
\end{proof}
Both \textsf{FWSS} and Problem  \ref{def:SingleEdge} are also polynomially verifiable, hence NP-complete.
\newline
The classic dynamic programming approach that typically yields a polynomial time algorithm when profits are fixed is not efficient since in this variation we cannot optimize using the minimum value at each step due to negative values. Instead, we present a fully polynomial time approximation scheme (FPTAS).

\begin{algorithm}[ht]
 \KwData{number of transactions $n$, values of the sequence of transactions $v_i \in \mathbb{R}, \forall 1 \leq i \leq n$, capital $C$, approximation factor $\epsilon$.}
 \KwResult{binary vector $S=\{0,1\}^n$ that indicates which transactions to execute.}
 Let $K= \frac{\epsilon V}{n}$, where $V=\max_{1 \leq i \leq n}v_i$\;
 For all transactions $1 \leq i \leq n$ define ${v'}_i= \lfloor \frac{v_i}{K} \rfloor$\;
 Let $T(i,j)=0$, for all $1 \leq i \leq n$ and $1 \leq j \leq \frac{n^2}{\epsilon}$\;
 \For{$i=1$ to $n$}{
    \For{$j=1$ to $\frac{n^2}{\epsilon}$}{
    \begin{equation*} 
	T(i,j)  =
	\begin{cases} 
	\max \{T(i-1, j), 1+T(i-1,j-{v'}_i) \} & ,if \enspace   \frac{C}{K} \geq j-{v'}_i> 0 \\
	T(i-1, j) & ,else\\ 
	\end{cases}
	\end{equation*}
    Store for every $T(i,j)$ a $n$-binary vector $S_{i,j}$ that has value $1$ in the $k$-th position if the $k$-th transactions is chosen to be executed\;
    }
    }
Return vector $S_{i,j}$ for the maximum $T(i,j)$ such that $\sum_{k=1}^{n} S_{i,j}(k) \cdot v_k \leq C$\;
 \caption{\tt MaxProfit}
 \label{alg:maxprofit}
\end{algorithm}
\newpage

\begin{theorem}
Algorithm {\tt MaxProfit} is a fully polynomial time approximation scheme for Problem \ref{def:SingleEdge}.
\end{theorem}
\begin{proof}
The running time of the algorithm is $O(\frac{n^3}{\epsilon})$, which is polynomial in both $n$ and $\frac{1}{\epsilon}$.\\
We will prove that the profit of the output of algorithm {\tt MaxProfit} is at least $(1-\epsilon)$ times the optimal. We denote by $S$ the set of transactions returned by the algorithm, $O$ the set returning the optimal profit and $prof(X)$ the profit from the set of transactions $X$. Since we scaled down by $K$ and then rounded down, for every transaction $i$ we have that $K{v'}_i \leq v_i$. Therefore, the optimal set's profit can decrease at most $nK$, $prof(O)-prof'(O)K \leq nK$. The dynamic program returns the optimal set for the scaled instance. Thus, $prof(S) \geq prof'(O)K \geq prof(O) -nK = prof(O) - \epsilon V \geq (1-\epsilon) prof(O)$, since $prof(O) \geq V$.
\end{proof}
%
%
\textbf{\textsf{Scaling to many channels.}} 
Unfortunately, even when the graph is a tree, algorithm \ref{alg:maxprofit} does not scale efficiently. Creating an $m$-dimensional tensor for the dynamic program, where $m$ are the edges of the tree, has time complexity $O(C^mn)$ where $C$ is the maximum capital from all edges. Even if we bound the capital by a polynomial on $n$ the algorithm remains exponential due to the number of  edges on the exponent. In the general case where the graph could contain cycles, the problem becomes even more complex. Now, we need to additionally consider all possible routes for each transaction; this adds an exponential factor on the running time of the algorithm. \\
Since Problem \ref{def:capitalcon} is complex, we study special cases that might be useful in practice and provide an insight to the general problem.
%
%
\subsection{Channel Design for Maximum Profit} 
In this section, our goal is to find the minimum capital for which we can achieve maximum profit, i.e., execute all profitable transactions (Problem \ref{def:maxprofit}). 
At first, we note some simple observations for the graph structure. Then, we prove that any star graph is a $2$-approximate solution with respect to the capital, but even the ``best'' star is not an optimal solution. Last, we prove the problem is NP-hard when there are graph restrictions. \\
Throughout this section, we refer to the optimal solution of Problem \ref{def:maxprofit} as the \emph{optimal network for maximum profit}. 

\begin{lemma}\label{lem:nocycles}
When  the capital is unlimited, the optimal network for maximum profit does not contain cycles.
\end{lemma}
\begin{lemma}\label{lem:maxprofit}
When the capital is unlimited, there exists an algorithm, with time complexity $\Theta(n)$, where $n$ denotes the number of transactions, 
that returns the optimal network for maximum profit. 
\end{lemma}
\begin{lemma} \label{cc}
The optimal network for maximum profit is not necessarily a connected graph.
\end{lemma}
We refer to transactions that increase the PSP's profit as \textsf{profitable transactions}. We assume all nodes participate in at least two transactions (Lemma \ref{lem:twotx}).
\begin{lemma}\label{lem:proftx}
Not all transactions are \textsf{profitable transactions}.
\end{lemma}

Despite Lemma \ref{cc}, we note that payment channels are monetary systems. As such, large companies are expected to participate in the network as highly connected nodes, ensuring that the optimal graph is one connected component. Thus, for the rest of the section we can safely assume that the optimal graph is connected.

We will now define some formal notation to prove that choosing any star as the graph to route all transactions requires at most twice the capital of the optimal graph. This immediately implies we have a $2$-approximation to Problem \ref{def:maxprofit}.

Now, suppose we can update the capital of an edge before executing each transaction. This way we can guarantee there is enough capital on all channels for each transaction execution. These updates are for free, like assigning tokens, and we use them as a stepping stone to calculate the total capital (amortized analysis).
Let us denote $c_{G}(uv,i)$ the additional capital required at the edge $(u,v)$, for transaction $t_i$ with direction from $u$ to $v$ on graph $G$. Now, we have that the total capital on graph $G$, denoted by $C_G$, is \[ C_{G}= \sum_{\forall (u,v) } \sum_{\forall i} c_{G}(uv,t_i) \]
Moreover, let $opt$ denote the optimal graph and $V$ the set of nodes involved in $opt$.\\
We will show that the capital used to route a sequence of transactions on any star that contains the same set of nodes as the optimal graph is at most twice the capital used by the optimal solution for the same sequence.

\begin{theorem}
Any star graph yields a $2$-approximate solution for Problem \ref{def:maxprofit}.
\end{theorem}
\begin{proof} To prove the theorem, we just need to prove that for any sequence of transactions $t_1, t_2, \dots, t_n$, for any star graph $S(V)$, $C_S \leq 2C_{opt}$.
We will show that we can execute on the star graph the same sequence of transaction as the optimal solution with twice as many tokens (amortized capital).
Initially we have zero tokens on all edges on both the optimal and the star graph. Every time a new transaction $t_i$ comes the optimal solution finds a path from sender to receiver. For every edge $(u,v)$ on this path the optimal solution assigns $c_{opt}(uv,t)$ tokens. Then,  we assign on the star, $S$, $c_{opt}(uv,t)$ tokens on the edges $mu$ and $vm$, where $m$ is the central node on $S$. The only exceptions are the sender and receiver nodes, $s$ and $r$ respectively, where the tokens are initially placed on $sm$ and $mr$ to execute the transaction. Thus, for every transaction the sum of the tokens used on the star graph are twice the sum of the tokens used on the optimal solution. Therefore, the overall required capital on the star is at most twice the optimal capital, $C_S \leq 2C_{opt}$. \\
To complete our proof, we need to show we assigned in total enough tokens to execute the given sequence of transactions. When a new transaction comes from $s$ to $t$, we only need to guarantee there enough tokens on $sm$ and $mt$. 
Obviously, if a transaction needs additional tokens to be executed on the optimal graph then the aforementioned strategy guarantees the additional tokens for the star graph as well.
If there are already some tokens on the optimal graph for the sender then either he was previously an intermediate node or a receiver node. In both those cases the same amount of tokens would have been stored on $sm$ as well. With a similar argument, if there were some tokens for the last edge to reach the receiver on the optimal graph then $r$ was either an intermediate node or a sender. Again, in both those cases the same amount of tokens would have been assigned to $mr$ on the star. 
\end{proof}

\begin{lemma}\label{lem:starnotopt}
The star graph is not an optimal solution for Problem \ref{def:maxprofit}.
\end{lemma}

\noindent\textbf{\textsf{Discussion.}} 
The centralized nature of the star is quite convenient for a payment network operated by a PSP. The star alleviates the problem of participation incentives detected on decentralized payment networks; now the participants of the network can be online only when they want to execute a transaction. 
Although the star graph is not optimal, it is a good enough solution for a PSP, since the capital he needs to lock in the channels is at most twice the minimum. Thus, payment hubs are an economically viable solution for the throughput problem on cryptocurrencies. 
%
%
\subsection{Channel Design with Graph Restrictions}
An interesting variation of the problem is when the network has restrictions (Problem \ref{def:subgraph}). Instead of allowing all possible channels, we assume some of them cannot occur in real life. In this case, we are given a graph with all the potential channels, the sequence of transactions and the capital, and we want to find the induced subgraph that maximizes the profit. 
We prove that the problem of deciding whether all given transactions can be executed in the given graph with a fixed capital is NP-complete.

The graph is given so the capital needed to open the channels is fixed in each given instance. Thus, we assume the capital corresponds solely to the capital we lock on the edges but not the one we require to open the channels. 
\begin{theorem}\label{thm:graphrest}
Problem \ref{def:subgraph} is NP-complete.
\end{theorem}
\section{Online Channel Design}
In this section, we study the online case, assuming no prior knowledge for the future transactions. When there is a transaction request we instantly decide whether to execute it or not through our network, assuming we have enough capital on the edges of the path we want to route the transaction. If there is not enough capital on some of the edges, we can refund a channel, which costs 1,  the same as opening a new channel.
%
%
\subsection{Single Channel with Capital Constraints}
Similarly to the offline case, we first focus on the simpler case where we have a single edge and limited capital. The transactions arrive online, for each transaction we immediately decide whether it is accepted.

\begin{theorem}
There is no competitive algorithm for adaptive adversaries.
\end{theorem}
\begin{proof}
Suppose we have a channel with $\cright = \cleft = 5$. Transactions from left to right have positive values, those from right to left have negative values.
Let us consider two different transaction sequences:
\begin{enumerate}
\item $(1,5,-10,10,-10,10,\ldots)$
\item $(1,4,-10,10,-10,10,\ldots)$
\end{enumerate}
Apart from the second transaction, both sequences are identical: The first transaction has value $1$, starting with the third transaction we always move the complete capital with every transaction. The only difference is the second transaction. 

If some online algorithm accepts the first transaction, then the adversary presents the first sequence; if the online algorithm denies the first transaction, then the adversary reveals the second sequence. Therefore, no matter whether this online algorithm accepts the first transaction or not, it can at most accept one transaction, while the optimal offline algorithm can accept almost all transactions (in case of the first sequence, the offline algorithm only needs to deny the first transaction, in case of the second sequence it will accept all transactions). 
\end{proof}
 
%
%
\subsection{Channel Design for Maximum Profit}
We assume again that we want to execute all transactions, thus the optimal graph does not contain cycles. Our objective is to minimize the capital, given all transactions will be executed through our payment network. Wlog, we assume the PSP is a node in the network. Similarly to the offline case, we show that constructing a star network to connect the nodes with payment channels is a good solution. Specifically, we present a log-competitive algorithm that takes advantage of the star graph structure. \vspace{0.3cm}

\begin{algorithm}[h] \label{alg:online}
 \KwData{online sequence of transactions $t_i=(s_i,r_i,v_i)$}
 \KwResult{capital $C$}
 We denote by $s$ the node corresponding to the PSP.\\
 $E \leftarrow \emptyset$\\
 $C \leftarrow 0$\\
 \For{each transaction $t_i$}{
    \uIf{$s_i$ is not connected to $s$}{
    $E \leftarrow E \cup (s_i, s)$ \\
    $c_{s_i,s} \leftarrow v_i$,
    $c_{s,s_i} \leftarrow v_i$\\
    $C \leftarrow C+1$
    }
    \uElseIf{$c_{s_i,s} < v_i$}{
    $c_{s_i,s} \leftarrow c_{s_i,s} + v_i$\\
    $c_{s,s_i} \leftarrow c_{s,s_i} + v_i$\\
    $C \leftarrow C+1$
    }
    \Else{
    $c_{s_i,s} \leftarrow c_{s_i,s} -v_i$\\
    $c_{s,s_i} \leftarrow c_{s,s_i} + v_i$
    }
    For the case of $r_i$ we follow a similar (invert) procedure.
    }
    \For{all $i \neq s$}{
    $C \leftarrow C + c_{i,s} + c_{s,i}$
    }
Return capital $C$\\
 \caption{\tt OnlineMaxProfit}
\end{algorithm}

In Algorithm \texttt{OnlineMaxProfit}, we gradually form a star where the center is the PSP. At each step, we check whether there is enough capital on the edges to and from the center to execute the current transaction. If the capital on an edge is smaller that the value of the current transaction, we refund the channel and add to the capacity of this edge twice the value of the current transaction.

\begin{theorem}
Algorithm \texttt{OnlineMaxProfit} is $\Theta(\log C_{opt})$-competitive.
\end{theorem}
\begin{proof}
The star is a $2$-approximation to the optimal offline solution, thus we start with a competitive ratio of $2$. The way we update the capacities, each time adding twice the value of the transaction if the capacity is less than the transaction's value, yields also a competitive ratio of two on the edges' capacities. Moreover, at each such step we at least double the capacity of an edge thus we reach the edge's optimal capital, $C_{e}$, in $\log C_{e}$ steps. If we sum over all edges, in total we refund the channels at most $(n-1)\log C_{edges}$ times, where $n$ is the number of nodes in the network and $C_{edges}$ the edges' optimal capital of the offline solution. Therefore, algorithm \texttt{OnlineMaxProfit} returns $C\leq (n-1)\log C_{edges} + 4C_{edges}$, while the offline solution requires $C_{opt} = (n-1) + C_{edges}$. This yields a competitive ratio of $\Theta (\log C_{opt})$.
\end{proof}

\section{Conclusion} 
We introduced a graph theoretic framework for payment networks. We studied the problem for a specific epoch, i.e., for a fixed number of transactions. This restriction is due to privacy issues, such as timing attacks on the payment network that can leak information on the customers' personal data.
We tried to maximize the profit (the number of accepted transactions minus the number of generated channels) and to minimize the capital needed to execute these transactions. Due to the multi-objective nature, there are several versions of this problem. 
In this paper, we mainly focused on two interesting variations: 
\begin{enumerate}
\item How to choose transactions to execute on a single channel with given capital assignments to maximize the profit, 
\item How to design a network and assign capitals to accept all transactions and minimize the needed capital. 
\end{enumerate}
It turns out, these two problems are challenging, as we show that the first problem and a variation of the second one are both NP-hard. 
We propose a dynamic programming based algorithm for the single channel problem and show that it is an FPTAS. 
For the network design and capital assignment problem, 
we show that stars achieve approximation ratio $2$. In other words, hubs are not only an implementable and privacy-guaranteed solution, as mentioned in \cite{heilman2017tumblebit} and \cite{green2017bolt}, but also a satisfactory solution for PSP from the profit-maximization point of view. 

We also studied the online versions of these problems. For the single channel case we show that it is impossible to design a competitive algorithm against an adaptive adversary. For the online channel design for maximum profit, we devise an $O(\log C)$-competitive online algorithm based on the star structure. 

The results presented in this paper and the proposed algorithms can be applied to other fields such as traffic network design. For example, every airline would want to maximize the profit and to minimize the costs (of creating new routes and purchasing new airplanes). Interestingly, similar to what we discovered, hubs are indeed used by almost all airlines, e.g., most flights of the Turkish airline departure from or fly to Istanbul. 

Apart from capital assignment, fee assignment of payment networks \cite{avarikioti2018channelfees} is also related to the traffic network design problem. 
One need to pay for using highways in some countries (e.g., Greece, China and France), thus the companies need to decide which cities are connected by highways and how much one needs to pay for every path. In this way, the drivers prefer highways (analog to the payment channels) to other slow paths (analog to the main chain), and hence the profit is maximized. 
%
%

%
\bibliography{references}

\begin{thebibliography}{10}

\bibitem{ethereum}
Ethereum white paper.
\newblock URL: \url{https://github.com/ethereum/wiki/wiki/White-Paper}.

\bibitem{raiden2017}
Raiden network.
\newblock 2017.
\newblock URL: \url{http://raiden.network/}.

\bibitem{avarikioti2018channelfees}
Georgia Avarikioti, Gerrit Janssen, Yuyi Wang, and Roger Wattenhofer.
\newblock Payment network design with fees.
\newblock 2018.
\newblock URL:
  \url{https://github.com/zetavar/Payment-Network-Design-with-Fees/blob/master/Payment_Network_Design_with_Fees-Full_Version.pdf}.

\bibitem{back2014sidechains}
Adam Back, Matt Corallo, Luke Dashjr, Mark Friedenbach, Gregory Maxwell, Andrew
  Miller, Andrew Poelstra, Jorge Timón, and Pieter Wuille.
\newblock Enabling blockchain innovations with pegged sidechains.
\newblock 2014.
\newblock URL: \url{https://www.blockstream.com/sidechains.pdf}.

\bibitem{burchert2017multipartychannels}
Conrad Burchert, Christian Decker, and Roger Wattenhofer.
\newblock {Scalable Funding of Bitcoin Micropayment Channel Networks}.
\newblock In {\em {19th International Symposium on Stabilization, Safety, and
  Security of Distributed Systems (SSS), Boston, Massachusetts, USA}}, November
  2017.

\bibitem{croman2016scaling}
Kyle Croman, Christian Decker, Ittay Eyal, Adem~Efe Gencer, Ari Juels, Ahmed
  Kosba, Andrew Miller, Prateek Saxena, Elaine Shi, Emin G{\"u}n~Sirer, Dawn
  Song, and Roger Wattenhofer.
\newblock On scaling decentralized blockchains.
\newblock In {\em Financial Cryptography and Data Security}, pages 106--125.
  Springer Berlin Heidelberg, 2016.

\bibitem{DW2015channels}
Christian Decker and Roger Wattenhofer.
\newblock A fast and scalable payment network with bitcoin duplex micropayment
  channels.
\newblock In Andrzej Pelc and Alexander~A. Schwarzmann, editors, {\em
  Stabilization, Safety, and Security of Distributed Systems}, pages 3--18,
  Cham, 2015. Springer International Publishing.

\bibitem{green2017bolt}
Matthew Green and Ian Miers.
\newblock Bolt: Anonymous payment channels for decentralized currencies.
\newblock In {\em Proceedings of the 2017 ACM SIGSAC Conference on Computer and
  Communications Security}, CCS '17, pages 473--489, 2017.

\bibitem{heilman2017tumblebit}
Ethan Heilman, Leen Alshenibr, Foteini Baldimtsi, Alessandra Scafuro, and
  Sharon Goldberg.
\newblock Tumblebit: An untrusted bitcoin-compatible anonymous payment hub.
\newblock In {\em Network and Distributed Systems Security Symposium 2017
  (NDSS), February 2017}.

\bibitem{hu1974ocst}
T.~Hu.
\newblock Optimum communication spanning trees.
\newblock {\em SIAM Journal on Computing}, 3(3):188--195, 1974.
\newblock URL: \url{https://doi.org/10.1137/0203015}, \href
  {http://arxiv.org/abs/https://doi.org/10.1137/0203015}
  {\path{arXiv:https://doi.org/10.1137/0203015}}, \href
  {http://dx.doi.org/10.1137/0203015} {\path{doi:10.1137/0203015}}.

\bibitem{johnson1978networkdesign}
D.~S. Johnson, J.~K. Lenstra, and A.~H.~G. Kan~Rinnooy.
\newblock The complexity of the network design problem.
\newblock {\em Networks}, 8(4):279--285.
\newblock URL:
  \url{https://onlinelibrary.wiley.com/doi/abs/10.1002/net.3230080402}, \href
  {http://arxiv.org/abs/https://onlinelibrary.wiley.com/doi/pdf/10.1002/net.3230080402}
  {\path{arXiv:https://onlinelibrary.wiley.com/doi/pdf/10.1002/net.3230080402}},
  \href {http://dx.doi.org/10.1002/net.3230080402}
  {\path{doi:10.1002/net.3230080402}}.

\bibitem{Karp1972}
Richard~M. Karp.
\newblock {\em Reducibility among Combinatorial Problems}, pages 85--103.
\newblock Springer US, Boston, MA, 1972.
\newblock URL: \url{https://doi.org/10.1007/978-1-4684-2001-2_9}.

\bibitem{kokoris2017omniledger}
Eleftherios Kokoris-Kogias, Philipp Jovanovic, Linus Gasser, Nicolas Gailly,
  Ewa Syta, and Bryan Ford.
\newblock Omniledger: A secure, scale-out, decentralized ledger via sharding.
\newblock 2017.

\bibitem{luu2016sharding}
Loi Luu, Viswesh Narayanan, Chaodong Zheng, Kunal Baweja, Seth Gilbert, and
  Prateek Saxena.
\newblock A secure sharding protocol for open blockchains.
\newblock In {\em Proceedings of the 2016 ACM SIGSAC Conference on Computer and
  Communications Security}, pages 17--30. ACM, 2016.

\bibitem{malavolta2017silentwhispers}
Giulio Malavolta, Pedro Moreno-Sanchez, Aniket Kate, and Matteo Maffei.
\newblock Silentwhispers: Enforcing security and privacy in decentralized
  credit networks.
\newblock In {\em Network and Distributed Systems Security Symposium 2017
  (NDSS)}.

\bibitem{nakamoto2008bitcoin}
Satoshi Nakamoto.
\newblock Bitcoin: A peer-to-peer electronic cash system.
\newblock 2008.

\bibitem{poon2015lightning}
Joseph Poon and Thaddeus Dryja.
\newblock The bitcoin lightning network: Scalable off-chain instant payments.
\newblock 2015.
\newblock URL: \url{https://lightning.network}.

\bibitem{prihodko2016flare}
Pavel Prihodko, Slava Zhigulin, Mykola Sahno, Aleksei Ostrovskiy, and Olaoluwa
  Osuntokun.
\newblock Flare: An approach to routing in lightning network.
\newblock 2016.
\newblock URL:
  \url{https://bitfury.com/content/downloads/whitepaper_flare_an_approach_to_routing_in_lightning_network_7_7_2016.pdf}.

\bibitem{roos2018routing}
Stefanie Roos, Pedro Moreno-Sanchez, Aniket Kate, and Ian Goldberg.
\newblock Settling payments fast and private: Efficient decentralized routing
  for path-based transactions.
\newblock In {\em Network and Distributed Systems Security Symposium 2018
  (NDSS)}.

\end{thebibliography}
\newpage
\appendix

\section{Omitted Proofs} 

\begin{proof}[Proof of Lemma \ref{lem:twotx}]
Assume node $u$ is in the graph and sends and receives less than two transactions. Since $u$ is part of the graph, it has at least one neighbor. Choose an arbitrary neighbor $v$ of $u$, connect all remaining neighbors of $u$ directly with $v$ (if not already connected), and then remove $u$. For each neighbor $w$ of node $u$, increase the capital of edge $(w,v)$ in the new graph by the capital of the removed edge $(w,u)$ (on both sides of the edge). Now all transactions routed originally through edge $(w,u)$ can be routed in the new graph through edge $(w,v)$, the new total capital needed is at most the same as the old capital, and there is enough capital to route all previously routable transactions.
We denote by $opt$ the profit of the optimal solution. The graph without $u$ has at least one edge less. Thus, the profit of the new graph for (at least) the same set of transactions is at least $opt-(1-\epsilon)+1 > opt$, since setting up the channel (edge) $(u,v)$ costs $1$ but $u$'s transaction fees are at most $1-\epsilon$. So we are better off without node $u$. 
\end{proof}

\begin{proof}[Proof of Lemma \ref{lem:tree}]
Suppose we are given the following sequence of transactions and capital $C=6a+6$ where $a>11$:\\
$t_i=(v_2,v_1,1), \mbox{ for} \ 1 \leq i \leq a$\\
$t_i=(v_2,v_3,1),  \mbox{ for}\ a+1 \leq i \leq 2a$\\
$t_i=(v_4,v_3,1),  \mbox{ for}\ 2a+1 \leq i \leq 3a$\\
$t_i=(v_4,v_5,1), \mbox{ for}\ 3a+1 \leq i \leq 4a$\\
$t_i=(v_6,v_5,1), \mbox{ for}\ 4a+1 \leq i \leq 5a$\\
$t_i=(v_6,v_1,1),  \mbox{ for}\ 5a+1 \leq i \leq 6a$\\
For this example, we will show that the optimal solution that maximizes the profit given the capital $C=6a+6$ returns a graph that contains a cycle. 

One solution is the graph $G(V,E)$, where $V=\{v_1,v_2,v_3,v_4,v_5,v_6\}$, $E=\{(v_1,v_2),(v_2,v_3),$ $(v_3,v_4),(v_4,v_5),(v_5,v_6),(v_6,v_1)\}$. The profit in $G$ is $6a-6$, since all $6a$ transactions are executed and connect directly and $6$ channels are opened. The spent capital is $6a+6 = C$ since we open $6$ edges and lock capital $a$ on the sender side of each edge. 

If the graph is not connected we lose at least $a$ transactions and remove at most $6$ edges. So, the total profit decreases by at least $a-6 > 5$, hence the optimal graph is connected. 

Suppose now the optimal graph does not contain any cycle. Since the optimal graph is connected, it is a spanning tree. Let's denote by $\ell$ the number of leaves in the tree with $2\le \ell \le 5$, since the spanning tree has $5$ edges. If we want to have at least the profit of the cycle described above, we must deliver at least $6a-1$ transactions, since we only saved the opening cost of a single edge.

We can see that the capital locked on the edges connecting these leaves is at least $2a\ell-2$, since, for every node, both transactions involved are either outgoing or incoming. For every other edge, the capital locked on it is at least $a-1$. Therefore, the total number of locked capital is at least $2a\ell-2 + (a-1)(5-\ell) \ge 7a-5 > 6a + 6 = C$. 
\end{proof}

\begin{proof}[Proof of Lemma \ref{lem:FWSS}]
We will reduce \textsf{Subset Sum (SS)} \cite{Karp1972} to \textsf{FWSS}. \\
\textsf{SS}: Given a set of integers $U = \{a_1, a_2, \ldots, a_n \}$, and integer $A$, is there a non-empty subset $U' \subseteq U$ such that $\sum_{a_i \in U'} a_i = A$? \\
Given an instance of \textsf{SS}, we define $n$ different instances of \textsf{FWSS,} one for each possible value of $l$, where the set of integers $U$ and the integer value $A$ are the same for every \textsf{FWSS} instance as in \textsf{SS}. If one of the \textsf{FWSS} instances returns ``yes'' then we return ``yes'', else we return ``no''. 
If any instance of \textsf{FWSS} returns ``yes'', then the same subset satisfies the \textsf{SS} problem, thus it must return ``yes''.
If all instances of \textsf{FWSS} return ``no'', then there is no set satisfying the \textsf{SS} problem, since we checked all possible set sizes. Thus, \textsf{SS} must return ``no'' as well. The transformation is polynomial to the input.
\end{proof}

\begin{proof}[Proof of Lemma \ref{lem:nocycles}]
Suppose the optimal graph contains a cycle. We choose an arbitrary edge of the cycle, denoted $e=(i,j)$, where $i,j$ the nodes of edge $e$. We remove edge $e$ and reroute all transactions using edge $e$ through the path from $i$ to $j$. This path exists since removing an edge from a cycle cannot disconnect the graph. The profit of the graph without $e$ is the profit of the optimal graph plus $1$ (the cost of opening edge $e$). This is a contradiction, since the optimal network maximizes the profit.
\end{proof}

\begin{proof}[Proof of Lemma \ref{lem:maxprofit}]
We present the following algorithm:
\begin{enumerate}
\item Traverse the input transactions and create a list $L$ that contains the number of transactions between every two nodes (potential edges).
\item Traverse list $L$ as follows:\\
If the number of transactions is at least $2$ and adding the edge between the two nodes does not form a cycle, add the edge to the (initially empty) graph.
\end{enumerate}
The algorithm guarantees there is a path between every two nodes that want to execute at least two transactions (Lemma \ref{lem:twotx}). Thus, all transactions that increase the PSP's profit are executed. Moreover, the algorithm does not contain a cycle, so the output graph is minimal regarding the edges. 
The first argument maximizes $|S|$ and the second minimizes $|E|$ for the given sequence of transactions. Therefore, the algorithms returns a graph that achieves maximum profit ($max\ |S| -|E|$).\\
The running time of the algorithm is linear to the number of transactions, $\Theta(n)$.
\end{proof}

\begin{proof}[Proof of Lemma \ref{cc}]
Suppose the graph of the optimal solution is always connected. 
Assume now we are given the following sequence of transactions: $t_1=(v_1,v_2,1), t_2=(v_1,v_2,1), t_3=(v_3,v_4,1), t_4=(v_3,v_4,1), t_5=(v_1,v_3,1)$. Since the optimal graph for this example is connected, there are at least three edges in the graph. Thus, the optimal profit is $|S|-|E|\leq 5(1-\epsilon)-3=2-5\epsilon$. However, if we consider the graph with edges $(v_1,v_2)$ and $(v_3,v_4)$, the PSP's profit is $2-4\epsilon$, greater than the profit of the optimal solution, which is a contradiction.
\end{proof}

\begin{proof}[Proof of Lemma \ref{lem:proftx}]
Suppose all transactions are \textsf{profitable transactions}. Lets assume we are given the following sequence of transactions: $t_1=(v_1,v_2,1), t_2=(v_1,v_2,1), t_3=(v_3,v_4,1), t_4=(v_3,v_4,1), t_5=(v_1,v_3,1)$. It is straightforward to see that any solution that executes all transactions must return a connected graph. However, we showed in the 
proof of Lemma \ref{cc} that the optimal graph for this example is not connected. Thus, there are transactions that are not executed in the optimal solutions and hence they are not \textsf{profitable transactions}.
\end{proof}

\begin{proof}[Proof of Lemma \ref{lem:starnotopt}]
We present a sequence of transactions for which any star graph requires larger capital to execute all transactions than the optimal solution, as illustrated in Figure \ref{fig:starnotopt}. Moreover, finding the optimal solution is not trivial in our example, even though we only consider unitary transactions.\\
The sequence of transactions is the following
$$ t_1=(v_4,v_2,1), t_2=(v_4,v_7,1), t_3=(v_2,v_6,1), t_4=(v_5,v_2,1), t_5=(v_6,v_5,1),$$ 
$$ t_6=(v_4,v_5,1), t_7=(v_3,v_2,1), t_8=(v_2,v_1,1), t_9=(v_1,v_3,1), t_{10}=(v_3,v_4,1) $$
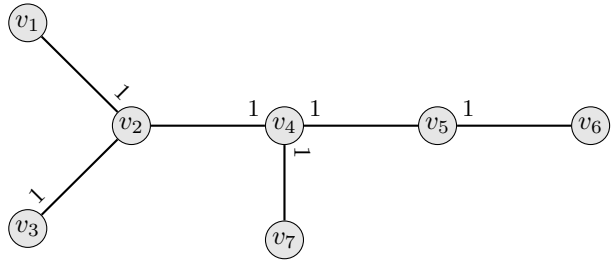
\begin{figure}
\begin{center}
 \begin{tikzpicture}
    \node[main node] (1) {$v_1$};
    \node[main node] (2) [below right = 1cm and 1cm of 1]  {$v_2$};
    \node[main node] (3) [below left = 1cm and 1cm of 2] {$v_3$};
    \node[main node] (4) [right =  1.5cm of 2] {$v_4$};
    \node[main node] (5) [right = 1.5cm of 4] {$v_5$};
    \node[main node] (6) [right = 1.5cm of 5] {$v_6$};
    \node[main node] (7) [below = 1cm of 4] {$v_7$};

    \path[draw,thick]
    (1) edge node[pos=0.9,sloped,above] {\small$1$} (2)
    (2) edge node[pos=0.9,sloped,above] {\small$1$} (3)
    (2) edge node[pos=0.9,sloped,above] {\small$1$} (4)
    (4) edge node[pos=0.1,sloped,above] {\small$1$} (7)
    (4) edge node[pos=0.1,sloped,above] {\small$1$} (5)
    (5) edge node[pos=0.1,sloped,above] {\small$1$} (6);
\end{tikzpicture}
\end{center}
\caption{The optimal graph and capital assignment to execute all transactions in the given sequence. The capital locked on each channel is illustrated by the number ($1$) above each edge and the position indicates the direction (in the initial state). }\label{fig:starnotopt}
\end{figure}
Figure \ref{fig:starnotopt}  illustrates the optimal graph (not necessarily unique); the capital needed to execute all transactions is $6$, which is the least possible since any tree with seven nodes has six edges and all of them are used at least one time with unitary values (in this example). There are seven different stars, one for each node as a center. It is easy to see that the capital needed for each one of them is at least $7$, which is strictly larger than the optimal.
\end{proof}

\begin{proof}[Proof of Theorem \ref{thm:graphrest}] 
We will reduce \textsf{Partition}, a known NP-complete problem \cite{Karp1972}, to Problem \ref{def:subgraph}.\\
\textsf{Partition:} Given a finite set $A=\{a_1, a_2, \dots, a_m\}$ and a size $s(a_i)\in \mathbb{Z}^+$ for each $a_i \in A$, is there a subset $A' \subseteq A$ such that $\sum_{a_i \in A'} s(a_i) = \sum_{a_i \in A-A'} s(a_i)$?\\
Given an instance of \textsf{Partition} we define an instance of Problem \ref{def:subgraph} as follows:\\
Let $v = \sum_{a_i \in A} s(a_i)$.
We consider the graph $G(V,E)$ with four nodes $V=\{b,c,d,e\}$ and edges $E=\{eb, bd, dc, ec \}$. We define the capital $C=2v$ and the sequence of $n=m+4$ transactions $$\{(d,b,v/2), (b,e,v/2), (d,c,v/2), (c,e,v/2), (e,d,s(a_1)), (e,d,s(a_2)), \dots, (e,d,s(a_n))\}$$
We will prove that this instance of \textsf{Partition} is a ``yes'' instance if and only if all transactions can be executed in $G$ with the given capital $C$. 
We denote by $T$ the last $m$ transactions defined above. 
Suppose all the transactions in the instance we defined can be executed using the capital $C$. 
This implies we can divide $T$ into two groups $T', T-T'$ that each sum to at most $C/2$ and route each group through one of the two paths from $e$ to $d$ in $G$. Since $\sum_{t_i \in T} v_i = C/2$, we have $|T'|=|T-T'|= C/4$. 
We define the set $A'=\{a_{i-4}|t_i \in T'\} \subseteq A$. It holds that $$\sum_{a_i \in A'} s(a_i) = \sum_{t_i \in T'} v_i = C/4 = \sum_{t_i \in T-T'} v_i =
\sum_{a_i \in A-A'} s(a_i)$$
For the opposite direction, suppose we cannot execute all transactions in $G$ with capital $C$. Since the first four transactions can always be executed with capital $C$, we cannot execute all transactions in $T$ with capital $C$. Thus, in the optimal solution we cannot partition the transactions in $T$ in two groups $T', T-T'$ with equal sum. Since we defined the values of transactions in $T$ to be the sizes of the elements in $A$, we conclude there is no subset $A' \subseteq A$ such that $\sum_{a_i \in A'} s(a_i) = \sum_{a_i \in A-A'} s(a_i)$.
\end{proof}

\end{document}